\documentclass[letter, 10pt, journal]{ieeetran}  
\IEEEoverridecommandlockouts                              
\usepackage{color}
\usepackage{psfrag}
\usepackage{balance}
\usepackage[hidelinks]{hyperref}
\usepackage{mathrsfs}
\usepackage{graphicx}
\usepackage{amssymb,amsfonts}

\usepackage{amsthm}
\usepackage{mathalfa}
\usepackage{tikz}
\usetikzlibrary{calc}
\usepackage{color}
\usepackage{pgfplots}
\pgfplotsset{compat=newest}
\usetikzlibrary{plotmarks}
\usepgfplotslibrary{patchplots}
\usepackage{grffile}
\usepackage{amsmath}
\newtheorem{rem}{Remark}
\newenvironment{remark}
  {\pushQED{\qed}\rem}
  {\popQED\endrem}

\newtheorem{fact}{Fact}
\newtheorem{proposition}{Proposition}
\newtheorem{problem}{Problem}
\newtheorem{theorem}{Theorem}

\newtheorem{definition}{Definition}

\usepackage{graphicx}
\usepackage{mathrsfs}
\usepackage{amsmath,amssymb,amsfonts,yhmath}
\usepackage{caption}
\usepackage{subcaption}
\usepackage{fancyhdr}
\usepackage{color}
\usepackage{tikz}
\usetikzlibrary{circuits.ee.IEC}
\usetikzlibrary{arrows}
\newcommand{\dz}{\operatorname{dz}}
\newcommand{\Xz}{\partial_z x}
\newcommand{\Xhatz}{\partial_z \hat{x}}
\newcommand{\Xhatt}{\partial_t \hat{x}}
\newcommand{\Xt}{\partial_t x}

\newcommand{\Sy}{\mathsf{S}}
\newcommand{\Di}{\mathsf{D}}

\newcommand{\Dpn}{\Di_+^{n_x}}

\newcommand{\R}{\mathfrak{R}}

\newcommand{\Z}{\mathbb{Z}}
\newcommand*{\tr}{^{\mkern-1.5mu\mathsf{T}}}

\newcommand{\dom}{\operatorname{dom}}
\newcommand{\Sign}{\operatorname{sign}}
\newcommand{\Id}{\mathbf{I}}

\newcommand{\He}{\operatorname{He}}

\newcommand{\Int}{\operatorname{int}}

\newcommand{\0}{{\bf 0}}

\newcommand{\Hone}{\mathcal{H}^1(0,1;\R^{n_x})}
\newcommand{\Ltwo}{\mathcal{L}^2(0,1; \R^{n_x})}

\newcommand{\ep}{{\varepsilon}}

\newcommand{\source}{{THIS IS A PREPRINT VERSION. IF YOU FOUND THIS READING USEFUL FOR YOUR RESEARCH PLEASE CITE THE PUBLISHED VERSION DOI: \href{https://doi.org/10.23919/ACC45564.2020.9147240}{https://doi.org/10.23919/ACC45564.2020.9147240}}}

\usepackage{mathtools}


\makeatletter
\def\ps@IEEEtitlepagestyle{}
\title{\LARGE \bf Observer Design for Systems of  Conservation Laws with Lipschitz Nonlinear Boundary Dynamics}\author{Francesco~Ferrante and Andrea~Cristofaro
	\thanks{F. Ferrante is with Univ. Grenoble Alpes, CNRS, Grenoble INP, GIPSA-lab, 38000 Grenoble, France. Email: francesco.ferrante@gipsa-lab.fr.}
\thanks{A. Cristofaro is with Department of Computer, Control and Management Engineering, Sapienza University of Rome, 00185 Rome, Italy. Email: cristofaro@diag.uniroma1.it}
\thanks{Research by F. Ferrante has been partially supported by the Grenoble Institute of Technology under the grant CrYStAL.}
\thanks{\textcolor{blue}{In this file, we propose a clearer proof sketch for Proposition 1 that clarifies some ambiguities in the published version. \textbf{The result is not affected}. In addition, we fixed some typos.} 
}
}
\begin{document}
\maketitle
\begin{abstract}
The problem of state estimation for a system of coupled hyperbolic PDEs and ODEs with Lipschitz nonlinearities with boundary measurements is considered. An infinite dimensional observer with a linear boundary injection term is used to solve the state estimation problem. The interconnection of the observer and the system is written in estimation error coordinates and analyzed as an abstract dynamical system. The observer is designed to achieve global exponential stability of estimation error  with respect to a 
suitable norm. Sufficient conditions in the form of matrix inequalities are proposed  to design the observer. Numerical simulations support and corroborate the theoretical results.
\end{abstract}

\section{Introduction}
Many distributed physical phenomena can be described through mathematical models based on hyperbolic partial differential equations and conservation laws. Examples of such distributed systems can be found in hydraulic networks \cite{dos2008boundary}, multiphase flow \cite{di2011slugging}, transmission networks \cite{fliess1999active}, road traffic networks \cite{hante2009modeling} or gas flow in pipelines \cite{dick2010classical}. The interest in boundary controllers and boundary observers is motivated by the fact that systems governed by partial differential equations are typically equipped with sensors located at the boundary of the spatial domain, while state variables are not directly measured in the interior. In this regard, conditions for stability, controllability and
observability of first-order hyperbolic systems have been largely explored; see for instance \cite{coron2007strict,krstic2008backstepping,prieur2008robust,tucsnak2009observation,li2010controllability} and the references therein. More specifically,  observer design for linearized first-order hyperbolic systems based on Lyapunov methods has been addressed in \cite{aamo2006observer}, where exponential convergence is guaranteed using boundary injections. For quasilinear first-order hyperbolic systems, the tasks of boundary stabilization and state estimation have been considered in \cite{coron2013local}.

When dynamic boundary conditions are present, a coupled system of PDEs and ODEs arises, this making the structure of the problem particularly rich and interesting. Systems modeled as the interconnection of PDEs and ODEs can be found in numerous applications; see, e.g., \cite[Chapter 1]{bastin:coron:book:2016}, \cite{de2018backstepping,roman2019robustness}.
An interesting approach for stability analysis of the interconnection of a linear ODE and a wave equation is presented in \cite{barreau2018lyapunov} where cross-terms defined through supplementary integral states are introduced. Along the same lines, in \cite{trinh2017design} non-diagonal Lyapunov functionals are considered for stability analysis of coupled systems of scalar PDEs and ODEs. 
The design of boundary observer for linear and quasilinear conservation laws with static and asymptotically stable dynamic boundary control is proposed in \cite{castillo2013boundary}, where the authors state sufficient conditions based on matrix inequalities.  A design procedure for backstepping observers is proposed in \cite{hasan2016boundary}, where the solution of an auxiliary set of PDEs is use to determine a suitable change of coordinates. The problem of designing high-gain observers for hyperbolic systems of balance laws, yet with distributed in-domain measurements, has been addressed in \cite{kitsos2018high,kitsos2019high}.

Recently, in \cite{ferrante2019boundary}, we proposed a design technique for boundary observers for linear hyperbolic systems with lower-order reaction terms and potentially unstable dynamic boundary conditions based on matrix inequalities. In this paper we
extend the constructions in \cite{ferrante2019boundary} to design an exponentially convergent observer for a system of linear conservation laws with Lipschitz nonlinear boundary dynamics; a similar setting is considered in \cite{ahmed2018pde}. The class of Lipschitz functions is recognized to include some of the most common nonlinearities in control systems such as saturations or dead-zones, this motivating the interest in such a generalized setup. The design of asymptotic observers for finite-dimensional Lipschitz nonlinear systems via the use of matrix inequalities is a well-established research area; see \cite{zemouche2013lmi, alessandri2004design,arcak2001nonlinear} just to mention a few. To address the state estimation problem considered in this paper, we consider an infinite-dimensional observer that is a copy of the system with a linear boundary injection term that needs to be designed. Then, using an integral Lyapunov functional and some upper growth bounds for the nonlinearity, we provide sufficient conditions in the form of matrix inequalities to enforce exponential decay of the state estimation error. 

The remainder of the paper is organized as follows.  Section~\ref{sec:Prob_Statement} describes the considered setup and establishes some preliminary results on the estimation error dynamics. 
Section~\ref{sec:Main} pertains to stability analysis of the error dynamics and provides sufficient conditions for observer design. Section~\ref{sec:Num_example} illustrates the application of the proposed observer in a numerical example. Finally, some conclusions are reported in Section \ref{sec:Conclusions}.
\subsection*{{\bf Notation}}
The sets $\R_{\geq 0}$ and $\R_{>0}$ represent the set of nonnegative and positive real scalars, respectively. The symbols $\Sy^n$ ($\Sy_+^n$) and $\Di_+^n$ denote, respectively, the set of real $n\times n$ symmetric (symmetric positive definite) matrices and the set of diagonal positive definite matrices. For a matrix $A\in\R^{n\times m}$, $A\tr$ denotes the transpose of $A$, and $\He(A)= A+A\tr$. For a symmetric matrix $A$, positive definiteness (negative definiteness) and positive semidefiniteness (negative semidefiniteness) are denoted, respectively, by $A\succ 0$ ($A\prec 0$) and $A\succeq 0$ ($A\preceq 0$). 
In partitioned symmetric matrices, the symbol $\bullet$ stands for symmetric blocks. The symbol $\Id$ denotes the identity matrix. Let $X\subset\R^n$, 
$Y\subset\R$, $x\in X$, and $f\colon X \rightarrow Y$, the symbol 
$\nabla f(x)$ denotes the gradient of $f$ at $x$.
For a vector $x\in\R^n$, $\vert x \vert$ denotes its Euclidean norm. Given $x, y\in\R^n$, we denote by $\langle x, y \rangle_{\R^n}$ the standard Euclidean inner product.
Let $X$ and $Y$ be linear normed spaces and $f\colon X\rightarrow Y$, we denote by $Df(x)$ the Fr\'echet derivative of $f$ at $x$. Let $k=0, 1,\dots $, the symbol $\mathcal{C}^{k}(X, Y)$ denotes the set of class $\mathcal{C}^{k}$ functions $f\colon X\rightarrow Y$.
Let $U\subset\R$ and $V\subset\R^n$, and $f, g\colon U\rightarrow V$. 
We denote by $\langle f, g\rangle_{\mathcal{L}^2(U, V)}=\int_U f(x)g(x) dx$, and
$\Vert f\Vert_{\mathcal{L}^2(U, V)}=\sqrt{\langle f, f\rangle_{\mathcal{L}^2(U, V)}}$. Let $U\subset\R$ be open and $V$ be a linear normed space,  
$$
\begin{aligned}
\!\mathcal{H}^1(U;V)\!\coloneqq\!&\left\{\!f\!\in\!\mathcal{L}^2(U;V)\colon\!\!f\,\text{is absolutely continuous on}\,\,U, \right.\\
&\left.\frac{d}{dz} f\in\mathcal{L}^2(U;V)\right\}
\end{aligned}
$$
where $\frac{d}{dz}$ stands for the weak derivative of $f$. Let $X$ be a normed space, $\mathscr{A}\subset X$ be closed, and $x\in X$, the distance of $x$ to $\mathscr{A}$  is defined as $\vert x\vert_{\mathscr{A}}\coloneqq\inf_{y\in\mathscr{A}}\Vert x-y\Vert_X$.
\section{Problem Formulation and Outline of the Solution}
\label{sec:Prob_Statement} 
Let $\Omega\coloneqq(0,1)$, we consider a system of $n_x$ linear $1$-D hyperbolic PDEs with dynamic boundary conditions formally written as:
\begin{equation}
\label{eq:hyp_PDEs_3}
\def\arraystretch{1.2}
\begin{array}{llll}
&\displaystyle{\Xt(t,z)+\Lambda \Xz(t,z)=0}&(t,z)\in\R_{>0}\!\times\Omega\\
&x(t, 0) = C\chi(t)&\forall t\in\R_{>0}\\
&\dot{\chi}(t)= A\chi(t)+B\Psi(Z\chi(t))&\forall t\in\R_{>0}\\
&y(t)=Mx(t, 1)&\forall t\in\R_{>0}
\end{array}
\end{equation}
where  $\Xt$ and $\Xz$ denote, respectively, the derivative of $x$ with respect to ``time'' and the ``spatial'' variable $z$, $(z\mapsto x(\cdot, z), \chi)\in\Ltwo\times\R^{n_\chi}$ is the system state, $y\in\R^{n_y}$ is a measured output, and $\Psi\colon\R^{n_z}\rightarrow\R^{n_l}$ is globally 
Lipschitz continuous with Lipschitz constant $\ell>0$. Namely, for all $z_1, z_2\in\R^{n_z}$ one has
\begin{equation}
\vert\Psi(z_1)-\Psi(z_2)\vert\leq \ell \vert z_1-z_2\vert
\end{equation}
Matrices   
$\Lambda\in\Dpn$, $B\in\R^{n_\chi\times n_l}$,  $C\in\R^{n_x\times n_\chi}, A\in\R^{n_\chi\times n_\chi}$,  
$Z\in\R^{n_z\times n_\chi}$, and $M\in\R^{n_y\times n_x}$ are known. 
Our goal is to design an observer providing an exponentially converging estimate $(\hat{x}(\cdot, z), \hat{\chi})$ of the system state $(x(\cdot, z), \chi)$ from any initial condition.
Inspired by \cite{castillo2013boundary, ferrante2019boundary}, we consider the following observer with state $(\hat{x}, \hat{\chi})\in\Ltwo\times\R^{n_\chi}$:
\begin{equation}
\label{eq:observer}
\def\arraystretch{1.2}
\begin{array}{llll}
&\displaystyle{\Xhatt(t,z)+\Lambda \Xhatz(t,z)=0}\\
&\hat{x}(t, 0) = C\hat{\chi}(t)\\
&\dot{\hat{\chi}}(t)= A\hat{\chi}(t)+B\Psi(Z\hat{\chi}(t))+L(y(t)-M\hat{x}(t, 1))\\
&&\hspace{-1.8cm}\!\!(t,z)\!\in\!\R_{>0}\!\!\times\!\Omega\\
\end{array}
\end{equation}
where $L\in\R^{n_\chi\times n_y}$ is the observer gain to be designed. At this stage, define the following invertible change of variables
$$
\begin{aligned}
&\ep\coloneqq x-\hat{x}\\
&\eta\coloneqq \chi-\hat{\chi}
\end{aligned}
$$ 
which defines the two estimation errors. Then, the interconnection of 
observer \eqref{eq:observer} and plant \eqref{eq:hyp_PDEs_3} reads as follows
\begin{equation}
\label{eq:CL}
\def\arraystretch{1.2}
\begin{array}{llll}
&\displaystyle{\Xt(t,z)+\Lambda \Xz(t,z)=0}&\\
&x(t, 0) = C\chi(t)\\
&\dot{\chi}(t)= A\chi(t)+B\Psi(Z\chi(t))\\
&\displaystyle{\partial_t \ep(t, z)+\Lambda \partial_z \ep(t, z)=0}\\
&\ep(t, 0) = C\eta(t)\\
&\dot{\eta}(t)= A\eta-L M\ep(t, 1)+B\rho(\chi(t), \eta(t))\\
&&\hspace{-1.5cm}\!\!(t, z)\!\in\!\R_{>0}\!\!\times\!\Omega
\end{array}
\end{equation}
where for all $(\chi, \eta)\in\R^{2n_\chi}$:
\begin{equation}
\label{eq:rhodef}
\rho(\chi, \eta)\coloneqq \Psi(Z\chi)-\Psi(Z(\chi-\eta))
\end{equation}
In view of the global Lipschitzness assumption on $\Psi$,  it easily follows that for all $(\chi, \eta)\in\R^{2n_\chi}$
\begin{equation}
\label{eq:rhoBound}
\rho(\chi, \eta)\tr \rho(\chi, \eta)-\ell^2 \eta\tr Z\tr Z\eta\leq 0
\end{equation}
\subsection{Abstract Formulation}
Similarly as in \cite{bastin:coron:book:2016}, in this paper, we focus on mild solutions to \eqref{eq:CL}. To this end, as in \cite{CurtainZwart:95,arendt2011vector, jacob2015c}, we reformulate \eqref{eq:CL} as an abstract differential equation on the Hilbert space $\mathcal{Z}\coloneqq(\Ltwo\times\R^{n_\chi})^2$ 
endowed with the following standard inner product:
$$
\langle (a_1, a_2), (b_1, b_2)\rangle_{\mathcal{Z}}\coloneqq\langle a_1, a_2\rangle_{\sqrt{\mathcal{Z}}}+\langle a_2, b_2\rangle_{\sqrt{\mathcal{Z}}}
$$
where for all $a_1=(a_1^x, a_1^\chi), a_2=(a_2^x, a_2^\chi)\in\Ltwo\times\R^{n_\chi}$ we define: 
$$
\langle a_1, a_2\rangle_{\sqrt{\mathcal{Z}}}\coloneqq\langle a_1^x, a_2^x\rangle_{\Ltwo}+\langle a_1^\chi, a_2^\chi\rangle_{\R^{n_{\chi}}}
$$
In particular, let $\mathcal{S}\coloneqq\Hone\times\R^{n_\chi}$
and $\mathcal{X}\coloneqq\mathcal{S}^2\subset\mathcal{Z}$.  
Define
$$
\begin{aligned}
&\mathcal{D}_1\coloneqq\left\{(x, \chi)\in\mathcal{S}\colon x(0)= C\chi\right\}\\
&\mathcal{D}_2\coloneqq\left\{(\ep, \eta)\in\mathcal{S}\colon \ep(0)=C\eta\right\}
\end{aligned}
$$
and consider the following operator
\begin{equation}
\begin{aligned}
\mathcal{F}\colon&\dom\mathcal{F}\rightarrow\mathcal{Z}\\
&(x, \chi, \ep, \eta)\mapsto \begin{pmatrix}
\mathcal{A}_P&0\\
0&\mathcal{A}_O
\end{pmatrix} \begin{pmatrix}
x\\
\chi\\
\varepsilon\\
\eta
\end{pmatrix}
\end{aligned}
\label{eq:operator}
\end{equation}
for which $\dom\mathcal{F}=\mathcal{D}_1\times \mathcal{D}_2$, $\dom\mathcal{A}_P=\mathcal{D}_1$, $\dom\mathcal{A}_O=\mathcal{D}_2$, and
$$
\begin{aligned}
&\mathcal{A}_P\begin{pmatrix}
x\\
\chi
\end{pmatrix}\coloneqq\begin{pmatrix}
-\Lambda\frac{d}{dz}&0\\
0&A
\end{pmatrix}\begin{pmatrix}
x\\
\chi
\end{pmatrix}\\
&\mathcal{A}_O\begin{pmatrix}
\ep\\
\eta
\end{pmatrix}\coloneqq\begin{pmatrix}
-\Lambda\frac{d}{dz}&0\\
0&A
\end{pmatrix}\begin{pmatrix}
\ep\\
\eta
\end{pmatrix}+\begin{pmatrix}
0\\
-L M\ep(1) 
\end{pmatrix}
\end{aligned}
$$
Let $f\colon\R^{2n_\chi}\rightarrow\dom\mathcal{F}$ be defined for all $\pi=(\chi, \eta)\in\R^{2n_\chi}$ as:
$$
f(\pi)\coloneqq\begin{pmatrix}
0\\
B\Psi(\chi)\\
0\\
B\rho(\chi, \eta)
\end{pmatrix}
$$
Then, the error dynamics can be formally written as the following abstract differential equation on the Hilbert space $\mathcal{Z}$:
\begin{equation}
\begin{pmatrix}
\dot{x}\\
\dot{\chi}\\
\dot{\varepsilon}\\
\dot{\eta}
\end{pmatrix}=\mathcal{F}\begin{pmatrix}
x\\
\chi\\
\varepsilon\\
\eta
\end{pmatrix}+f(\chi, \eta)
\label{eq:abstract}
\end{equation}
In particular, following the lines of \cite{arendt2011vector,pazy2012semigroups}, we consider the following notion of (mild) solution for \eqref{eq:abstract}:
\begin{definition}
\label{def:Mild}
Let $\mathcal{I}\subset\R_{\geq 0}$ be an interval containing zero. A function $\varphi=(\varphi_x, \varphi_\chi, \varphi_\varepsilon, \varphi_\eta)\in\mathcal{C}^0(\mathcal{I}, \mathcal{Z})$ is a solution to \eqref{eq:abstract} if for all $t\in\mathcal{I}$
\begin{equation}
\label{eq:MildSol}
\begin{aligned}
&\int_0^t\varphi(s)ds\in\dom\mathcal{F}\\
&\varphi(t)=\varphi(0)+\mathcal{F}\int_0^t\varphi(s)+\int_0^t f(\varphi_\chi(s), \varphi_\eta(s))ds
\end{aligned}
\end{equation}
Moreover, we say that $\varphi$ is maximal if its domain cannot be extended and complete if $\mathcal{I}=[0, +\infty)$. 
\hfill$\circ$
\end{definition}
\begin{remark}
In \cite{pazy2012semigroups}, mild solutions to nonlinear evolution equations are defined by relying on the semigroup generated by the linear operator underlying the evolution equation. In this paper, inspired by \cite{arendt2011vector}, we follow a different approach and consider a definition of (mild) solution that is formulated directly on the data of the system, i.e., the operator $\mathcal{F}$ and the function $f$. On the other hand, it can be shown that when $\mathcal{F}$ generates a $\mathcal{C}_0$-semigroup on $\mathcal{Z}$ and and $f$ is Lipshchitz continuous the two definitions are equivalent.
\end{remark}
\medskip
\subsection{Existence and uniqueness of solutions}
In this subsection, we illustrate existence and uniqueness of the solutions to \eqref{eq:abstract}. 
\begin{proposition}[Existence, uniqueness, and regularity of solutions]
\label{prop:exist}
Let $\zeta=(x_0, \chi_0, \varepsilon_0, \eta_0)\in\mathcal{Z}$. 
Then, there exists a unique maximal solution
$\varphi=(\varphi_x, \varphi_\chi, \varphi_\varepsilon, \varphi_\eta)$ to \eqref{eq:abstract} such that $\varphi(0)=\zeta$. Furthermore $\varphi$ is complete. In addition, if $\zeta\in\dom\mathcal{F}$, then $\varphi$ is a classical solution to \eqref{eq:abstract}, i.e:
\begin{equation}
\label{eq:classic}
\begin{aligned}
&\varphi\in\mathcal{C}^1(\R_{>0}, \mathcal{Z})\\
&\varphi(t)\in\dom\mathcal{F}&\qquad\forall t\in\R_{>0}\\
&\dot{\varphi}(t)=\mathcal{F}\varphi(t)+f(\varphi_\chi(t), \varphi_\eta(t))&\qquad\forall t\in\R_{>0}
\end{aligned}
\end{equation}
\end{proposition}
\begin{proof}[Sketch of the proof]
\textcolor{blue}{As a first step, notice that $\mathcal{F}$ generates a $\mathcal{C}_0$-semigroup $\mathscr{T}$ on the Hilbert space $\mathcal{Z}$.
This can be proven, e.g., following an approach wholly similar as in the proof of \cite[Proposition 1]{ferrante2019boundary}. Using this fact, it can be easily shown that maximal solutions to \eqref{eq:abstract} correspond to maximal solutions to the following integral equation: 
\begin{equation}
\label{eq:IntSemi}
\varphi(t)=\mathscr{T}(t)\varphi(0)+\int_0^t \mathscr{T}(t-s)f(\varphi_\chi(s), \varphi_\eta(s))ds
\end{equation}
In particular, for all $\zeta\in\mathcal{Z}$, \cite[Theorem 1.2, page 184]{pazy2012semigroups} ensures that there exists a unique maximal solution to  \eqref{eq:IntSemi} such that $\varphi(0)=\zeta$ and that such a solution is complete. 
The last part of the statement, i.e., $\varphi$ satisfies \eqref{eq:classic} when
$\zeta\in\dom\mathcal{F}$, follows directly from \cite[Theorem 1.7, page 190]{pazy2012semigroups}.}
\end{proof}

Now we are in a position to state the problem we solve in this paper. To this end, since the objective is to ensure that the estimation error goes to zero as $t$ approaches infinity, we define the following closed subset of $\mathcal{Z}$:
\begin{equation}
\label{eq:attr}
\mathscr{A}\coloneqq\{(x, \chi, \ep, \eta)\in\mathcal{Z}\colon \ep=0, \eta=0\}
\end{equation} 
and observe that for all $\xi=(x, \chi, \ep, \eta)\in\mathcal{Z}$, one has
\begin{equation}
\label{eq:dist}
\vert \xi\vert_\mathscr{A}=\sqrt{\langle \ep, \ep\rangle_{\Ltwo}+\langle \eta, \eta\rangle_{\R^{n_\chi}}}
\end{equation}
\begin{problem}
\label{prob:ObserverDesign}
Given $A\in\R^{n_\chi\times n_\chi}, C\in\R^{n_x\times n_\chi}$, $M\in\R^{n_y\times n_x}$, $\Lambda\in\Di_+^{n_x}$, $B\in\R^{n_\chi\times n_l}$, $Z\in\R^{n_z\times n_\chi}$ and let $\mathscr{A}$ be defined as in \eqref{eq:attr}.
Design $L\in\R^{n_\chi\times n_y}$ such that each solution $\varphi$ to \eqref{eq:abstract} satisfies
\begin{equation}
\vert \varphi(t)\vert_{\mathscr{A}}\leq \kappa e^{-\lambda t}\vert \varphi(0)\vert_{\mathscr{A}}\quad \forall t\in\dom \varphi
\label{eq:GES}
\end{equation}
for some (solution independent) $\kappa, \lambda\in\R_{>0}$. \hfill$\circ$
\end{problem}
Notice that due to completeness of maximal solutions to \eqref{eq:abstract}, \eqref{eq:GES} ensures that the estimation error approaches zero exponentially along maximal solutions.
\section{Sufficient Conditions for Observer Design}
\label{sec:Main}
In this section we propose sufficient conditions for the solution to Problem~\ref{prob:ObserverDesign}. To this end, let us recall the following result from \cite{ferrante2019boundary} whose role will be clarified later in the proof of Theorem~\ref{theo:theo1}.
\begin{proposition}
Let $P_1\in\Di_+^{n_x}$, $P_2\in\R^{n_\chi \times n_x}$, $P_3\in\Sy_+^{n_\chi}$, and $\mu\in\R$. Define
\begin{equation}
\begin{aligned}
\mathbb{G}\colon&[0,1]\rightarrow\Sy^{n_x+n_\chi}\\
&z\mapsto\left[
\begin{array}{cc}
e^{-\mu z}P_1& P_2\tr\\
\bullet& P_3
\end{array}\right]
\end{aligned}
\label{eq:Gdef}
\end{equation}
For all $(\ep, \eta)\in\dom \mathcal{A}_O$ one has
\begin{equation}
\begin{aligned}
&2\int_0^1 \left\langle\mathbb{G}(z)\begin{bmatrix}\ep(z)\\ \eta\end{bmatrix}, \mathcal{A}_O\begin{bmatrix}\ep(z)\\ \eta\end{bmatrix}\right\rangle_{\R^{n_x+n_\chi}}dz=\\
&\displaystyle{\int_0^1}\begin{bmatrix}\ep(z)\\\ep(1)\\\eta\end{bmatrix}\tr\mathbb{M}(z)\begin{bmatrix}\ep(z)\\\ep(1)\\\eta\end{bmatrix}dz
\end{aligned}
\label{eq:InnerProdV}
\end{equation}
where for all $z\in[0, 1]$, $\mathbb{M}$ is defined in \eqref{eq:Mz} (at the top of the next page). 
\begin{figure*}
\vspace{0.5cm}

\begin{equation}
\scalebox{1}{$\mathbb{M}(z)=
\begin{bmatrix} 
-e^{-\mu z}\mu\Lambda P_1& -P_2\tr LM & P_2\tr A\\
\bullet&-\Lambda P_1 e^{-\mu z}&-\Lambda P_2\tr-M\tr L\tr P_3\\
\bullet& \bullet& \He(P_3A+P_2\Lambda C)+C\tr \Lambda P_1C
\end{bmatrix}$}
\label{eq:Mz}
\end{equation}
\end{figure*}
\label{prop:InnerProdV}
\hfill$\diamond$
\end{proposition}
The following fact somehow extends Proposition~\ref{prop:InnerProdV} to the setting considered in this paper.
\begin{fact}
Let $P_1\in\Di_+^{n_x}$, $P_2\in\R^{n_\chi \times n_x}$, $P_3\in\Sy_+^{n_\chi}$, $\mu\in\R$, and $\rho\colon\R^{2n_\chi}\rightarrow\R^{n_l}$ be defined as in \eqref{eq:rhodef}. Then, for all $(\ep,\chi,\eta)\in\Ltwo\times\R^{2n_\chi}$ one has:
\begin{equation}
\begin{aligned}
&2\int_0^1 \left\langle\mathbb{G}(z)\begin{bmatrix}\ep(z)\\\eta\end{bmatrix}, \begin{bmatrix}0\\B\rho(\chi, \eta)\end{bmatrix}\right\rangle_{\R^{n_x+n_\chi}}dz=\\
&\displaystyle{\int_0^1}\begin{bmatrix}\ep(z)\\\eta\\\rho(\chi, \eta)
\end{bmatrix}\tr
\begin{bmatrix}
0&0&P_2\tr B\\
\bullet&0&P_3B\\
\bullet&\bullet&0
\end{bmatrix}
\begin{bmatrix}\ep(z)\\ \eta\\\rho(\chi, \eta)\end{bmatrix}dz
\end{aligned}
\label{eq:InnerProdV}
\end{equation}
\label{fact:FactInner}
\hfill$\diamond$
\end{fact}
\begin{theorem}
\label{theo:theo1}
Assume that there exist $P_1\in\Di_+^{n_x}$, $P_2\in\R^{n_\chi\times n_x}$, $P_3\in\Sy^{n_\chi}$, $L\in\R^{n_\chi\times n_y}$, $\textcolor{blue}{\iota\geq 0}$, and $\mu\in\R_{>0}$ such that:
\begin{align}
& \left[
\begin{array}{cc}
P_1e^{-\mu}&P_2\tr\\
\bullet&P_3
\end{array}
\right]\succ 0
\label{eq:LMI1}\\
&\mathbb{K}\prec 0
\label{eq:LMI2}
\end{align}
where the matrix $\mathbb{K}$ is defined in \eqref{eq:Kappa}. Then, any maximal solution to \eqref{eq:abstract} satisfies
\eqref{eq:GES} with
\begin{equation}
\begin{aligned}
&\lambda=\frac{\alpha_3}{2\alpha_2},\quad \kappa=
\sqrt{\frac{\alpha_2}{\alpha_1}}
\end{aligned}
\label{eq:dissipation_quadratic}
\end{equation}
where 
\begin{equation}
\begin{aligned}
&\alpha_1\coloneqq\lambda_{\min}\left(\begin{bmatrix}P_1 e^{-\mu}&P_2\tr\\\bullet&P_3\end{bmatrix}\right),
&\alpha_2\coloneqq\lambda_{\max}\left(\begin{bmatrix}P_1&P_2\tr\\\bullet&P_3\end{bmatrix}\right)\\
&\alpha_3\coloneqq\vert\lambda_{\max}(\mathbb{K})\vert
\end{aligned}
\label{eq:alpha_12}
\end{equation}
\end{theorem}
\begin{proof}[Sketch of the proof]
Let $\mathbb{G}$ be defined as in \eqref{eq:Gdef}.
Consider the following Lyapunov functional candidate defined on $\Ltwo\times\R^{n_\chi}$:
\begin{equation}\label{eq:V}
V(\ep,\varphi)\coloneqq\displaystyle\int_0^1\left\langle\begin{bmatrix}\ep(z)\\ \varphi\end{bmatrix}, \mathbb{G}(z)\begin{bmatrix}\ep(z)\\ \varphi\end{bmatrix}\right\rangle_{\R^{n_x+n_\chi}} dz
\end{equation}
In particular, observe that for all $(x,\chi, \ep, \varphi)\in\mathcal{Z}$ one has
\begin{equation}
\label{eq:sandwhich}
\alpha_1 \vert (x,\chi, \ep, \varphi)\vert^2_\mathscr{A}
\leq V(\ep,\eta)\leq \alpha_2 \vert (x,\chi, \ep, \varphi)\vert^2_\mathscr{A}
\end{equation}
where $\alpha_1$ and $\alpha_2$ are strictly positive thanks to \eqref{eq:LMI1}. 
We show that the above results holds true for any classical solution to \eqref{eq:abstract}. The extension of the proof to mild solutions is discussed briefly at the end. 
Assume that $\varphi(0)\in\dom\mathcal{F}$. Then, from Proposition~\ref{prop:exist}: 
$\varphi\in\mathcal{C}^1(\Int\dom\varphi, \mathcal{Z})$, $\varphi(t)\in\dom\mathcal{F}$ for all $t\in\Int\dom\varphi$, and
\begin{equation}
\dot{\varphi}(t)=\mathcal{F}\varphi(t)+f(\varphi_\chi(t), \varphi_\eta(t))\quad \forall t\in\Int\dom\varphi
\label{eq:C1_sol}
\end{equation}
For all $t\in\Int\dom\varphi$, $\varphi$ being differentiable, one has
$$
\begin{aligned}
&\dot{V}(t)\coloneqq\frac{d}{dt}V(\varphi_\ep(t), \varphi_\eta(t))=DV(\varphi_\ep(t), \varphi_\eta(t))\begin{bmatrix}
\dot{\varphi}_\ep(t)\\ \dot{\varphi}_\eta(t)\end{bmatrix}=\\
&2\int_0^1\left\langle\mathbb{G}(z)\begin{bmatrix}
(\varphi_\ep(t))(z)\\ \varphi_\eta(t)\end{bmatrix}, \begin{bmatrix} (y_{\ep}(t))(z)\\  y_\eta(t)\end{bmatrix}\right\rangle_{\R^{n_x+n_\chi}} dz
\end{aligned}
$$
where for convenience we denoted: 
$$
\Ltwo\times\R^{n_\chi}\ni(y_{\ep}, y_\eta)\coloneqq\mathcal{A}_O\begin{bmatrix}\varphi_\ep(t)\\ \varphi_\eta(t)\end{bmatrix}
$$
Hence, by using Proposition~\ref{prop:InnerProdV} and Fact~\ref{fact:FactInner}, it follows that for all $t\in\Int\dom\varphi$
$$
\dot{V}(t)=\displaystyle{\int_0^1}\begin{bmatrix}(\varphi_\ep(t))(z)\\ (\varphi_\ep(t))(1)\\\varphi_\eta(t)\\
\rho(\varphi_\chi(t), \varphi_\eta(t))
\end{bmatrix}\tr\mathbb{H}(z)\begin{bmatrix}(\varphi_\ep(t))(z)\\ (\varphi_\ep(t))(1)\\\varphi_\eta(t)\\
\rho(\varphi_\chi(t), \varphi_\eta(t))
\end{bmatrix}dz
$$
where for all $z\in[0, 1]$ 
$$
\mathbb{H}(z)\coloneqq 
\begin{bmatrix}
\mathbb{M}(z)&\Gamma\\
\bullet&0\end{bmatrix}
$$
with 
$$
\Gamma\coloneqq\begin{bmatrix}P_2\tr B\\0\\P_3B\end{bmatrix}
$$
At this stage notice that in view of \eqref{eq:rhoBound}, for all $\iota\geq 0$ and all $t\in\Int\dom\varphi$, the following bound holds:
$$
\dot{V}(t)\leq \displaystyle{\int_0^1}\begin{bmatrix}(\varphi_\ep(t))(z)\\ (\varphi_\ep(t))(1)\\\varphi_\eta(t)\\
\rho(\varphi_\chi(t), \varphi_\eta(t))
\end{bmatrix}\tr\mathbb{Q}(z)\begin{bmatrix}(\varphi_\ep(t))(z)\\ (\varphi_\ep(t))(1)\\\varphi_\eta(t)\\
\rho(\varphi_\chi(t), \varphi_\eta(t))
\end{bmatrix}dz
$$ 
where for all $z\in[0, 1]$ 
$$
\mathbb{Q}(z)\coloneqq \begin{bmatrix}
\mathbb{M}(z)+\iota\ell\begin{bmatrix}
0\\0\\Z\tr\end{bmatrix}\begin{bmatrix}
0&0&Z\end{bmatrix}&\Gamma\\
\bullet&-\iota\Id\end{bmatrix} 
$$
Thus, since for all $z\in[0, 1]$, $\mathbb{Q}(z)\preceq\mathbb{K}$, by using \eqref{eq:LMI2} one has that for all $t\in\Int\dom\varphi $
$$
\begin{aligned}
\dot{V}(t)\leq &-\alpha_3 \int_0^1\left\vert \begin{bmatrix} (\varphi_\ep(t))(z)\tr&
(\varphi_\ep(t))(1)\tr&\varphi_\eta(t)\tr\end{bmatrix}\tr\right\vert^2dz
\\
\leq &-\alpha_3\vert \varphi(t)\vert^2_\mathscr{A} 
\end{aligned}
$$
the latter, thanks to \eqref{eq:sandwhich}, by standard manipulations yields for all $t\in\dom \varphi $
\begin{equation}
\label{eq:GES_bound_proof}
\vert \varphi(t)\vert_\mathscr{A} \leq \sqrt{\frac{\alpha_2}{\alpha_1}}e^{-\frac{\alpha_3}{2\alpha_2}t}\vert \varphi(0)\vert_\mathscr{A} 
\end{equation}
which reads as \eqref{eq:GES} with $\displaystyle\lambda=\frac{\alpha_3}{2\alpha_2}$ and $\kappa=\sqrt{\frac{\alpha_2}{\alpha_1}}$.  
The same bound can be established for solutions, defined as in Definition~\ref{def:Mild}, by relying on a density argument and on \eqref{eq:IntSemi}.
\end{proof}
\begin{figure*}
\vspace{0.2cm}

\begin{equation}
\scalebox{1}{$\mathbb{K}=
\begin{bmatrix} 
-e^{-\mu}\mu\Lambda P_1& -P_2\tr LM & P_2\tr A&P_2\tr B\\
\bullet&-\Lambda P_1 e^{-\mu }&-\Lambda P_2\tr-M\tr L\tr P_3&0\\
\bullet& \bullet& \He(P_3A+P_2\Lambda C)+C\tr \Lambda P_1C+\iota\ell^2 Z\tr Z&P_3B\\
\bullet&\bullet&\bullet&-\iota\Id
\end{bmatrix}$}
\label{eq:Kappa}
\end{equation}
\end{figure*}
\begin{remark}
Condition \eqref{eq:LMI2} is a bilinear matrix inequality (\emph{BMI}). As such, the numerical solution to \eqref{eq:LMI2} can be challenging when the size of the unknown matrices gets large. This is of course a limitation of the proposed approach and we are currently working to overcome this drawback. On the other hand, it is worth stressing that, as opposed to \cite{castillo2013boundary}, where conditions in the form of linear matrix inequalities are given for the design of observer \eqref{eq:observer}, the feasibility of \eqref{eq:LMI2} does not require asymptotic stability of the ODE dynamics. This aspect is connected to the selection of the non-block diagonal structure of the Lyapunov functional \eqref{eq:V}. Indeed, although such a functional enables to relax the assumptions on the ODE dynamics, it introduces some bilinear terms in \eqref{eq:LMI2}. 
\end{remark}
\section{Numerical Example}
\label{sec:Num_example}
Consider a system of the form \eqref{eq:hyp_PDEs_3} defined by the following data
$$
\begin{aligned}
&\Lambda=\begin{bmatrix}
\frac{3}{2}&0\\
0&2\\
\end{bmatrix}, A=\left[
\begin{array}{cr}
-1 & 2\\
2.05&-4
\end{array}
\right], B=\begin{bmatrix}
0\\\frac{1}{2}
\end{bmatrix}\\ 
&Z=\begin{bmatrix}
1&1
\end{bmatrix}, C=\Id,\quad M=\begin{bmatrix}
1&1
\end{bmatrix}, \Psi(z)=\dz(z)
\end{aligned}
$$
where for any $z\in\R$, the function $\dz\colon\R\rightarrow\R$ is defined as $\dz(z)=0$ if $\vert z\vert\leq 1$ and $\dz(z)=\Sign(z)(\vert z\vert-1)$ otherwise. Obviously, $\dz$ is Lipschitz continuous and, in particular, $\ell=1$.  
It is worth stressing that, while the PDE dynamics are stable, the boundary dynamics  are not, i.e. the matrix $A$ is not Hurwitz. Let us now design an observer according to the structure \eqref{eq:observer}. With the objective of searching for a feasible solution to \eqref{eq:LMI1}--\eqref{eq:LMI2}, we employed an heuristic algorithm wholly similar to \cite[Algorithm 1]{ferrante2019boundary}. In particular, for this example a feasible solution to \eqref{eq:LMI1}--\eqref{eq:LMI2} is as follows:
$$
\begin{array}{lll}
&P_1=\left[\begin{array}{cc} 11.76 & 0\\ 0 & 16.24 \end{array}\right]&P_2=\left[\begin{array}{cc} -6.904 & -7.157\\ -4.254 & -2.427 \end{array}\right]\\
&P_3=\left[\begin{array}{cc} 14.4 & 4.976\\ 4.976 & 7.52 \end{array}\right]&L=\left[\begin{array}{c} 0.4593\\ 0.2025\end{array}\right]\\
&\mu=0.4&\iota=3.335
\end{array}
$$
To validate our theoretical findings, next we present some simulations of the proposed observer\footnote{Numerical integration of hyperbolic PDEs is performed via the use of the Lax-Friedrichs (Shampine's two-step variant) scheme implemented in Matlab\textsuperscript{\tiny\textregistered} by Shampine \cite{shampine2005solving}.}.

In \figurename~\ref{fig:Lyapunov}, we report the evolution of the Lyapunov functional \eqref{eq:V}  (a logarithmic scale is employed on the $y$-axis) along the solution to \eqref{eq:abstract} from the following initial condition:
\begin{equation}
\begin{array}{lll}
&x_1(0, z)=\cos(2\pi z)&\forall z\in[0, 1]\\
&x_2(0, z)=-2\cos(4\pi z)&\forall z\in[0, 1]\\
&\chi(0)=(1, -2)\\
&\hat{x}(0, z)=0&\forall z\in[0, 1]\\
&\hat{\chi}(0)=(0, 0)
\end{array}
\label{eq:init_cond}
\end{equation}
\begin{figure}[!h]
\centering
\vspace{0.5cm}
\psfrag{z}[][][1]{$z$}
\psfrag{t}[][][1]{$t$}
\includegraphics[scale=0.4, trim={0.5cm 0cm 0.5cm 0.5cm},clip]{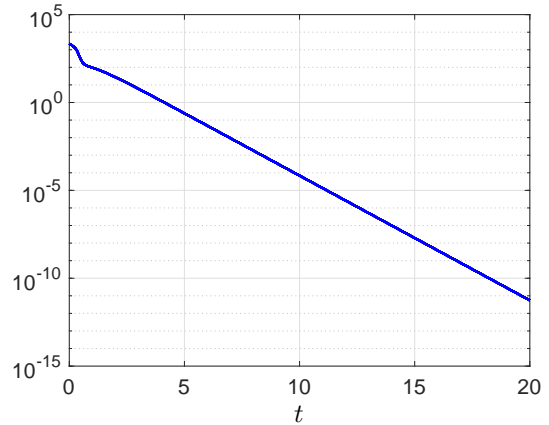}
\caption{Evolution of the Lyapunov functional \eqref{eq:V} (on a log-linear scale) along the solution to \eqref{eq:abstract} from the initial condition in \eqref{eq:init_cond}.}
\label{fig:Lyapunov}
\end{figure}
\figurename~\ref{fig:Lyapunov} clearly shows that $V$ converges exponentially to zero. Exponential state reconstruction is confirmed by \figurename~\ref{fig:error_PDE_example_3D}, and \figurename~\ref{fig:error_ODE_example} where the evolution of $\ep$ and $\eta$, respectively, are reported. 
\begin{figure}[!h]
\vspace{0.5cm}

\centering
\psfrag{e1}[][][1]{$\varepsilon_1$}
\psfrag{e2}[][][1]{$\varepsilon_2$}
\psfrag{t}[][][1]{$t$}
\psfrag{z}[][][1]{$z$}
\includegraphics[scale=0.4, trim={1.1cm 0.5cm 1.5cm 1.1cm},clip]{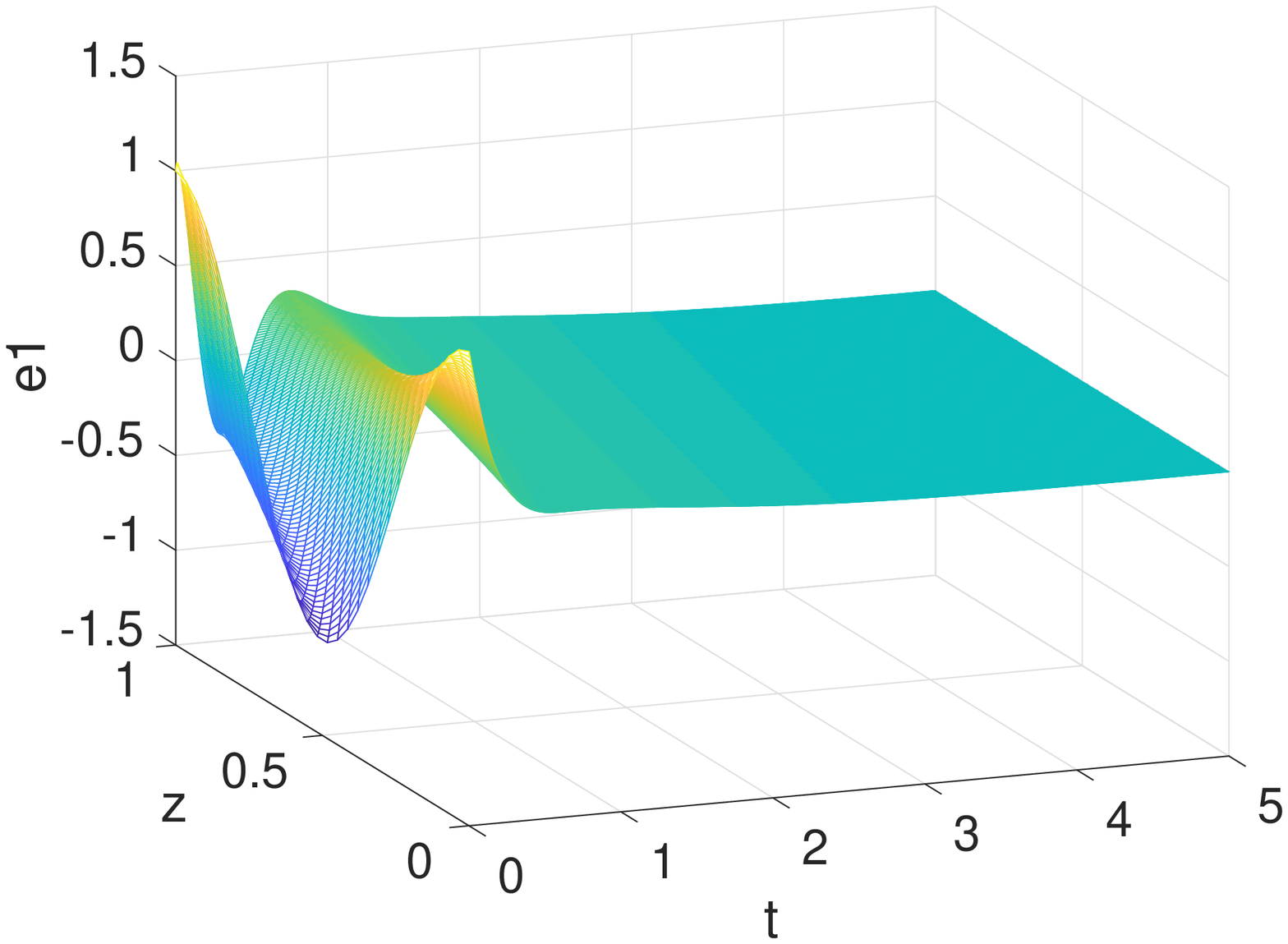}
\vspace{0.5cm}

\includegraphics[scale=0.4, trim={1.2cm 0.5cm 1.5cm 1.1cm},clip]{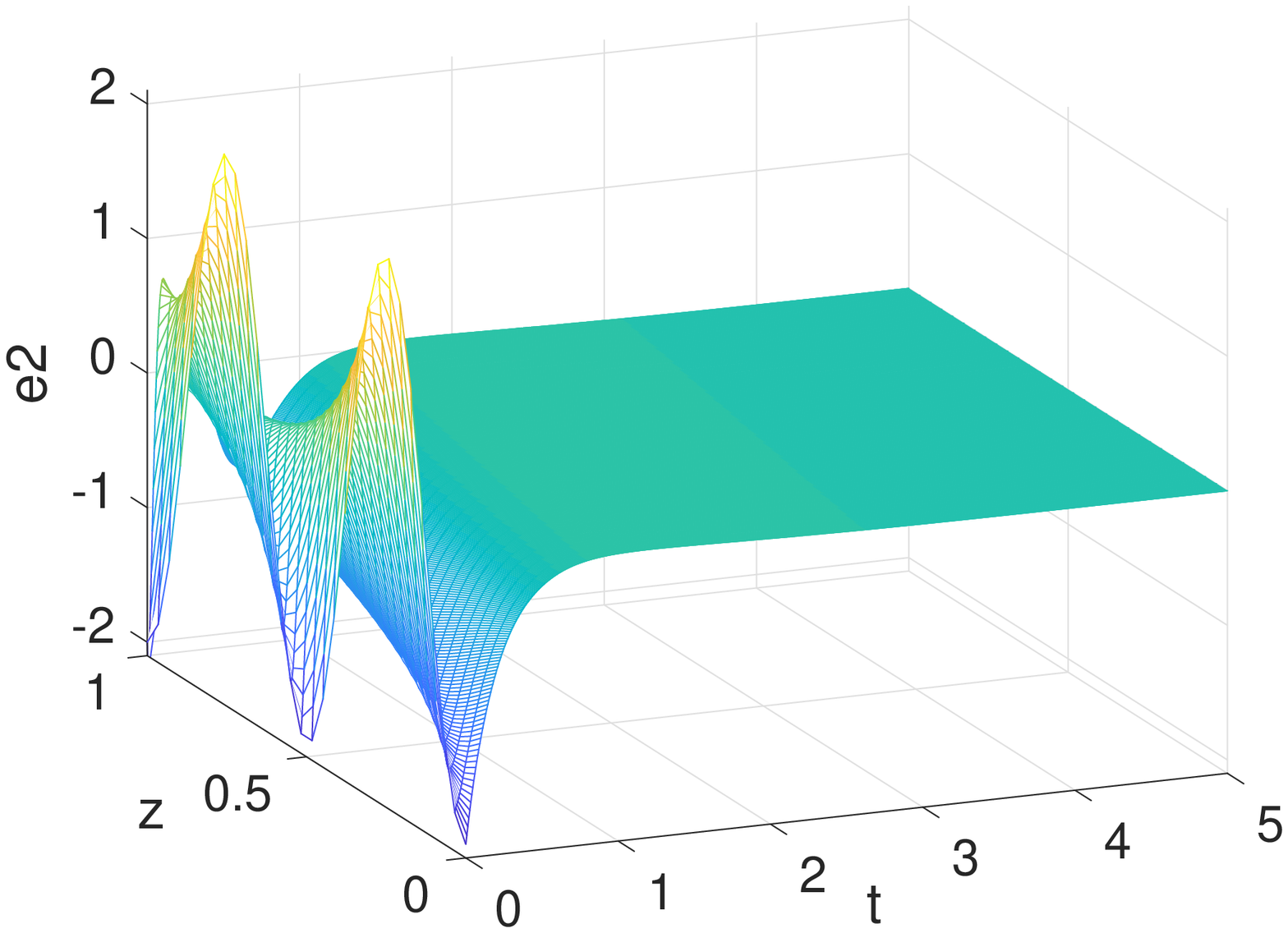}
\caption{Evolution of the error $\ep$ from the initial condition in \eqref{eq:init_cond}.}
\label{fig:error_PDE_example_3D}
\end{figure}
\begin{figure}[!h]
\centering
\psfrag{eta1}[][][1]{$\eta_1$}
\psfrag{eta2}[][][1]{$\eta_2$}
\psfrag{t}[][][1]{$t$}
\includegraphics[scale=0.43, trim={1cm 0cm 0cm 0cm},clip]{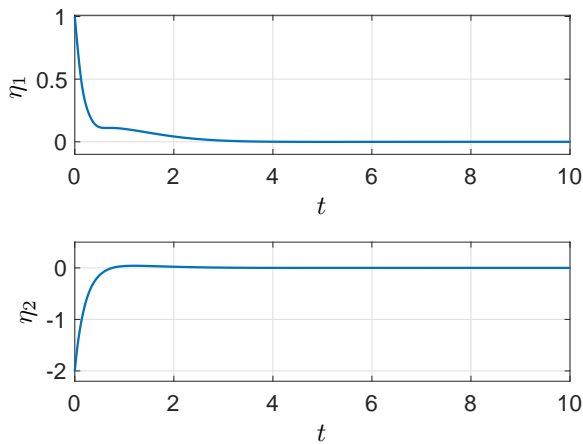}
\caption{Evolution of the error $\eta$ from the initial condition in \eqref{eq:init_cond}.}
\label{fig:error_ODE_example}
\end{figure}
\section{Conclusion}\label{sec:Conclusions}	
In this paper, we considered the problem of designing an observer to estimate the state 
of a system of linear conservation laws with Lipschitz nonlinear boundary dynamics. The observer we propose is a copy of the plant augmented with a linear output injection term. 
The interconnection of the plant and the observer is analyzed via abstract differential equations tools. The observer is designed to induce global exponential stability of a closed set in which the estimation error is equal to zero.
By pursuing a Lyapunov approach, the observer design problem is recast into the feasibility problem of some bilinear matrix inequalities. Numerical simulations are used to illustrate the effectiveness of the proposed observer design strategy in an example. 

Future research directions include the derivation of computationally affordable design algorithms for the observer based on linear matrix inequalities. 
\balance
\bibliographystyle{plain}
\bibliography{biblio}

\begin{thebibliography}{10}

\bibitem{aamo2006observer}
O.~M. Aamo, J.~Salvesen, and B.~A. Foss.
\newblock Observer design using boundary injections for pipeline monitoring and
  leak detection.
\newblock In {\em Proceedings of the International Symposium on Advanced
  Control of Chemical Processes}, pages 53--58, 2006.

\bibitem{ahmed2018pde}
T.~Ahmed-Ali, F.~Giri, M.~Krstic, and M.~Kahelras.
\newblock {PDE} based observer design for nonlinear systems with large output
  delay.
\newblock {\em Systems \& Control Letters}, 113:1--8, 2018.

\bibitem{alessandri2004design}
A.~Alessandri.
\newblock Design of observers for lipschitz nonlinear systems using lmi.
\newblock In {\em Proceedings of the 6th IFAC Symposium on Nonlinear Control
  Systems}, pages 459--464. Elsevier, 2004.

\bibitem{arcak2001nonlinear}
M.~Arcak and P.~Kokotovi{\'c}.
\newblock Nonlinear observers: a circle criterion design and robustness
  analysis.
\newblock {\em Automatica}, 37(12):1923--1930, 2001.

\bibitem{arendt2011vector}
W.~Arendt, C.~J.K. Batty, M.~Hieber, and F.~Neubrander.
\newblock {\em Vector-valued Laplace transforms and Cauchy problems},
  volume~96.
\newblock Springer Science \& Business Media, 2011.

\bibitem{barreau2018lyapunov}
M.~Barreau, A.~Seuret, F.~Gouaisbaut, and L.~Baudouin.
\newblock Lyapunov stability analysis of a string equation coupled with an
  ordinary differential system.
\newblock {\em IEEE Transactions on Automatic Control}, pages 3850--3857, 2018.

\bibitem{bastin:coron:book:2016}
G.~Bastin and J.-M. Coron.
\newblock {\em Stability and Boundary Stabilization of 1-{D} Hyperbolic
  Systems}, volume~88 of {\em Progress in Nonlinear Differential Equations and
  Their Applications}.
\newblock Springer, 2016.

\bibitem{castillo2013boundary}
F.~Castillo, E.~Witrant, C.~Prieur, and L.~Dugard.
\newblock Boundary observers for linear and quasi-linear hyperbolic systems
  with application to flow control.
\newblock {\em Automatica}, 49(11):3180--3188, 2013.

\bibitem{coron2007strict}
J.-M. Coron, B.~d'Andrea Novel, and G.~Bastin.
\newblock A strict {L}yapunov function for boundary control of hyperbolic
  systems of conservation laws.
\newblock {\em IEEE Transactions on Automatic control}, 52(1):2--11, 2007.

\bibitem{coron2013local}
J.-M. Coron, R.~Vazquez, M.~Krstic, and G.~Bastin.
\newblock Local exponential ${H}^2$-stabilization of a $2\times2$ quasilinear
  hyperbolic system using backstepping.
\newblock {\em SIAM Journal on Control and Optimization}, 51(3):2005--2035,
  2013.

\bibitem{CurtainZwart:95}
R.~Curtain and H.~Zwart.
\newblock {\em An introduction to infinite-dimensional linear systems theory}.
\newblock Springer-Verlag, 1995.

\bibitem{de2018backstepping}
G.~A. de~Andrade, R.~Vazquez, and D.~J. Pagano.
\newblock Backstepping stabilization of a linearized {ODE--PDE} rijke tube
  model.
\newblock {\em Automatica}, 96:98--109, 2018.

\bibitem{di2011slugging}
F.~Di~Meglio, G.~O. Kaasa, N.~Petit, and V.~Alstad.
\newblock Slugging in multiphase flow as a mixed initial-boundary value problem
  for a quasilinear hyperbolic system.
\newblock In {\em Proceedings of the American Control Conference}, pages
  3589--3596, 2011.

\bibitem{dick2010classical}
M.~Dick, M.~Gugat, and G.R. Leugering.
\newblock Classical solutions and feedback stabilization for the gas flow in a
  sequence of pipes.
\newblock {\em Networks \& Heterogeneous Media}, 5(4):691--709, 2010.

\bibitem{dos2008boundary}
V.~Dos~Santos and C.~Prieur.
\newblock Boundary control of open channels with numerical and experimental
  validations.
\newblock {\em IEEE transactions on Control systems technology},
  16(6):1252--1264, 2008.

\bibitem{ferrante2019boundary}
F.~Ferrante and A.~Cristofaro.
\newblock Boundary observer design for coupled {ODE}s-{H}yperbolic {PDE}s
  systems.
\newblock In {\em Proceedings of European Control Conference 2019}, pages
  2418--2423, 2019.

\bibitem{fliess1999active}
M.~Fliess, P.~Martin, N.~Petit, and P.~Rouchon.
\newblock Active signal restoration for the telegraph equation.
\newblock In {\em Proceedings of the 38th IEEE Conference on Decision and
  Control}, pages 1107--1111, 1999.

\bibitem{hante2009modeling}
F.~M. Hante, G.~Leugering, and T.~Seidman.
\newblock Modeling and analysis of modal switching in networked transport
  systems.
\newblock {\em Applied Mathematics and Optimization}, 59(2):275--292, 2009.

\bibitem{hasan2016boundary}
A.~Hasan, O.~M. Aamo, and M.~Krstic.
\newblock Boundary observer design for hyperbolic {PDE--ODE} cascade systems.
\newblock {\em Automatica}, 68:75--86, 2016.

\bibitem{jacob2015c}
B.~Jacob, K.~Morris, and H.~Zwart.
\newblock {$C_0$}-semigroups for hyperbolic partial differential equations on a
  one-dimensional spatial domain.
\newblock {\em Journal of evolution equations}, 15(2):493--502, 2015.

\bibitem{kitsos2018high}
C.~Kitsos, G.~Besan{\c{c}}on, and C.~Prieur.
\newblock High-gain observer design for a class of hyperbolic systems of
  balance laws.
\newblock In {\em 57th IEEE Conference on Decision and Control}, pages
  1361--1366. IEEE, 2018.

\bibitem{kitsos2019high}
C.~Kitsos, G.~Besan{\c{c}}on, and C.~Prieur.
\newblock A high-gain observer for a class of $2x2$ hyperbolic systems with
  $c1$ exponential convergence.
\newblock In {\em 3rd IFAC Workshop on Control of Systems Governed by Partial
  Differential Equations}, pages 174--179. Elsevier, 2019.

\bibitem{krstic2008backstepping}
M.~Krstic and A.~Smyshlyaev.
\newblock Backstepping boundary control for first-order hyperbolic {PDE}s and
  application to systems with actuator and sensor delays.
\newblock {\em Systems \& Control Letters}, 57(9):750--758, 2008.

\bibitem{li2010controllability}
D.~Li.
\newblock {\em Controllability and observability for quasilinear hyperbolic
  systems}.
\newblock American Institute of Mathematical Sciences, 2010.

\bibitem{pazy2012semigroups}
A.~Pazy.
\newblock {\em Semigroups of linear operators and applications to partial
  differential equations}, volume~44.
\newblock Springer Science \& Business Media, 2012.

\bibitem{prieur2008robust}
C.~Prieur, J.~Winkin, and G.~Bastin.
\newblock Robust boundary control of systems of conservation laws.
\newblock {\em Mathematics of Control, Signals, and Systems}, 20(2):173--197,
  2008.

\bibitem{roman2019robustness}
C.~Roman, D.~Bresch-Pietri, C.~Prieur, and O.~Sename.
\newblock Robustness to in-domain viscous damping of a collocated boundary
  adaptive feedback law for an anti-damped boundary wave {PDE}.
\newblock {\em IEEE Transactions on Automatic Control}, 64(8):3284--3299, 2019.

\bibitem{shampine2005solving}
L.~F. Shampine.
\newblock Solving hyperbolic {PDEs} in matlab.
\newblock {\em Applied Numerical Analysis \& Computational Mathematics},
  2(3):346--358, 2005.

\bibitem{trinh2017design}
N.-T. Trinh, V.~Andrieu, and C.-Z. Xu.
\newblock Design of integral controllers for nonlinear systems governed by
  scalar hyperbolic partial differential equations.
\newblock {\em IEEE Transactions on Automatic Control}, 62(9):4527--4536, 2017.

\bibitem{tucsnak2009observation}
M.~Tucsnak and G.~Weiss.
\newblock {\em Observation and control for operator semigroups}.
\newblock Springer Science \& Business Media, 2009.

\bibitem{zemouche2013lmi}
A.~Zemouche and M.~Boutayeb.
\newblock On {LMI} conditions to design observers for lipschitz nonlinear
  systems.
\newblock {\em Automatica}, 49(2):585--591, 2013.

\end{thebibliography}
\end{document}